\pdfoutput=1
\RequirePackage{ifpdf}
\ifpdf % We are running pdfTeX in pdf mode
\documentclass[pdftex]{sigma}
\else
\documentclass{sigma}
\fi

\numberwithin{equation}{section}

\newtheorem{Theorem}{Theorem}[section]

\begin{document}

\allowdisplaybreaks

\newcommand{\arXivNumber}{1704.07003}

\renewcommand{\thefootnote}{}

\renewcommand{\PaperNumber}{071}

\FirstPageHeading

\ShortArticleName{$N$-Bright-Dark Soliton Solution to a Vector SDNLS Equation}

\ArticleName{$\boldsymbol{N}$-Bright-Dark Soliton Solution to a Semi-Discrete\\ Vector Nonlinear Schr\"odinger Equation\footnote{This paper is a~contribution to the Special Issue on Symmetries and Integrability of Dif\/ference Equations. The full collection is available at \href{http://www.emis.de/journals/SIGMA/SIDE12.html}{http://www.emis.de/journals/SIGMA/SIDE12.html}}}

\Author{Bao-Feng FENG~$^\dag$ and Yasuhiro OHTA~$^\ddag$}

\AuthorNameForHeading{B.-F.~Feng and Y.~Ohta}

\Address{$^\dag$~School of Mathematical and Statistical Sciences, The University of Texas Rio Grande Valley,\\
\hphantom{$^\dag$}~Edinburg, TX 78539, USA}
\EmailD{\href{mailto:baofeng.feng@utrgv.edu}{baofeng.feng@utrgv.edu}}

\Address{$^\ddag$~Department of Mathematics, Kobe University, Rokko, Kobe 657-8501, Japan}
\EmailD{\href{mailto:ohta@math.kobe-u.ac.jp}{ohta@math.kobe-u.ac.jp}}

\ArticleDates{Received April 25, 2017, in f\/inal form September 03, 2017; Published online September 06, 2017}

\Abstract{In this paper, a general bright-dark soliton solution in the form of Pfaf\/f\/ian is constructed for an integrable semi-discrete vector NLS equation via Hirota's bilinear method. One- and two-bright-dark soliton solutions are explicitly presented for two-component semi-discrete NLS equation; two-bright-one-dark, and one-bright-two-dark soliton solutions are also given explicitly for three-component semi-discrete NLS equation. The asymptotic behavior is analysed for two-soliton solutions.}

\Keywords{bright-dark soliton; semi-discrete vector NLS equation; Hirota's bilinear method; Pfaf\/f\/ian}

\Classification{39A10; 35Q55}

\renewcommand{\thefootnote}{\arabic{footnote}}
\setcounter{footnote}{0}

\section{Introduction}
The nonlinear Schr\"{o}dinger (NLS) equation
 \begin{gather*} %\label{NLS}
{\mathrm i} u_{t}=u_{xx}+2 \sigma |u|^{2} u
\end{gather*}
is a generic model equation describing the evolution of small amplitude and slowly varying wave packets in weakly nonlinear media \cite{Ablowitzbook,APT, Agrawalbook,HasegawaKodama,AgrawalKivsharbook}. It arises in a variety of physical contexts such as nonlinear optics \cite{HasegawaTappert1, HasegawaTappert2}, Bose--Einstein condensates~\cite{BECReview}, water waves~\cite{Benney1967} and plasma physics~\cite{ZakharovPlasma}. The integrability, as well as the bright-soliton solution in the focusing case ($\sigma=1$), was found by Zakharov and Shabat~\cite{ZakharovShabat,ZakharovShsbat2}. The dark soliton was found in the defocusing NLS equation ($\sigma=-1$) \cite{HasegawaTappert2,Makhankov1982,Tsuzuki}, and was observed experimentally in~\cite{Krokeldark1,Weinerdark1}.

The integrable discretization of nonlinear Schr\"{o}dinger equation
\begin{gather*} %\label{AL}
\mathrm{i}q_{n,t}=\big( 1+\sigma |q_{n}|^{2}\big) (q_{n+1}+q_{n-1})
\end{gather*}
was originally derived by Ablowitz and Ladik \cite{AblowitzLadik,AL2}, so it is also called the Ablowitz--Ladik (AL) lattice equation. Similar to the continuous case, it is known that the AL lattice equation, by Hirota's bilinear method, admits the bright soliton solution for the focusing case ($\sigma =1$) \cite{Narita1990,TsujimotoBookchapter}, also the dark soliton solution for the defocusing case ($\sigma =-1)$~\cite{OhtaMaruno}. The inverse scattering transform (IST) has been developed by several authors in the literature \cite{ABB07,BPrinari,Mee15,Konotop92}.
The geometric construction of the AL lattice equation was given by Doliwa and Santini~\cite{DoliwaAL}.

The coupled nonlinear Schr\"{o}dinger equation
\begin{gather}
\begin{split}
& \mathrm{i}u_{t}=u_{xx}+2\big( \sigma _{1}|u|^{2}+\sigma_{2}|v|^{2}\big) u ,\\
& \mathrm{i}v_{t}=v_{xx}+2\big( \sigma _{1}|u|^{2}+\sigma_{2}|v|^{2}\big) v ,
\end{split} \label{CNLS}
\end{gather}
where \looseness=-1 $\sigma _{i}=\pm 1$, $i=1,2$, was f\/irstly recognized being integrable by Yajima and Oikawa~\cite{YajimaOikawa}. For the focusing-focusing case ($\sigma _{1}=\sigma _{2}=1$), the system~(\ref{CNLS}) solved by Manakov via inverse scattering transform (IST), admits the bright-bright soliton solution~\cite{Manakov72}, so it is also called the Manakov system in the literature. For the defocusing-defocusing case, the Manakov system admits bright-dark and dark-dark soliton solution \cite{PrinariMarkGino2006, Lakshmanan1995,Kivshar97}. However, the focusing-defocusing Manakov system admits all types of soliton solutions such as bright-bright solitons, dark-dark soliton, and bright-dark solitons \cite{KannaPRE2006, OhtaJianke, KannaPRE2008BrightDark}. The Manakov system can be easily extended to a~multi-component case, the so-called vector NLS equation. For the continuous vector NLS equation, the $N$-bright soliton solution was obtained in~\cite{Dubrovin1988vNLS,KannaPRE2006,YJZhang1994}; the general bright-dark and dark-dark soliton solutions were obtained in~\cite{Dubrovin1988vNLS,FengJPA2014, Makhankov1982,ParkShin2000, KannaPRE2008BrightDark}. The inverse scattering transform with nonvanishing boundary condition was solved by Prinari, Ablowitz and Biondini~\cite{PrinariMarkGino2006}. We should remark here that the problem of constructing exact soliton solutions to the vector NLS equation and proving their nonsingularity was settled by Dubrovin et al.\ in their landmark paper~\cite{Dubrovin1988vNLS}.

The semi-discrete coupled nonlinear Schr\"{o}dinger equation
\begin{gather}
\mathrm{i}q_{n,t}^{(1)} =\big( 1+\sigma _{1}\big|q_{n}^{(1)}\big|^{2}+\sigma _{2}\big|q_{n}^{(2)}\big|^{2}\big) \big( q_{n+1}^{(1)}+q_{n-1}^{(1)}\big) , \nonumber\\ \mathrm{i}q_{n,t}^{(2)} =\big( 1+\sigma _{1}\big|q_{n}^{(1)}\big|^{2}+\sigma _{2}\big|q_{n}^{(2)}\big|^{2}\big) \big( q_{n+1}^{(2)}+q_{n-1}^{(2)}\big),\label{sdCNLS}
\end{gather}
where $\sigma _{i}=\pm 1$, $i=1,2$, is of importance both mathematically and physically. It was solved by the inverse scattering transform (IST) in~\cite{TsuchidadCNLS, TsuchidasdCNLS}. The general multi-soliton solution in terms of Pfaf\/f\/ians was found by one of the authors recently \cite{OhtasdCNLS}, which is of bright type for the focusing-focusing case ($\sigma _{1}=\sigma _{2}=1$), and is of dark type for the defocusing-defocusing case ($\sigma _{1}=\sigma _{2}=-1$). However, as far as we know, no general mixed-type (bright-dark) soliton solution is reported in the literature, which motivated the present study.

In the present paper, we consider a $M$-coupled semi-discrete NLS equation of all types
\begin{gather}
\mathrm{i}q_{n,t}^{(j)}=\left( 1+\sum_{\mu=1}^{M}\sigma _{\mu}\big|q_{n}^{(\mu)}\big|^{2}\right) \big( q_{n+1}^{(j)}+q_{n-1}^{(j)}\big)
 ,\qquad j=1,2,\dots ,M, \label{MCNLS_sd}
\end{gather}
where $\sigma _{\mu}=\pm 1$ for $\mu=1,\dots M$. For all-focusing case ($\sigma _{\mu}=1$, $\mu=1,\dots, M$), its general $N$-bright soliton solution and the interactions among solitons were studied in~\cite{MarkOhta,OhtaBrightsdCNLS}. However, in contrast with a complete list of the general $N$-soliton solution to the vector NLS equation~\cite{FengJPA2014}, the mixed-type soliton solution of all possible nonlinearities (all possible values of~$\sigma_{\mu}$) is missing. The aim of the present paper is to construct a general $N$-bright-dark soliton solution to the semi-discrete vector NLS equation. The rest of the paper is organized as follows. In Section~\ref{section2}, we provide a general bright-dark soliton solution in terms of Pfaf\/f\/ians to the semi-discrete vector NLS equation (\ref{MCNLS_sd}) and give a rigorous proof by the Pfaf\/f\/ian technique \cite{HirotaBook,OhtaReview, OhtaDiscreteKP}. The one- and two-soliton solutions to two-coupled and three-coupled semi-discrete NLS equation are provided explicitly, respectively, in Section~\ref{section3}. We summarize the paper in Section~\ref{section4} and present asymptotic analysis for two-soliton solution in Appendix~\ref{appendixA}.

\section[General bright-dark soliton solution to semi-discrete vector NLS equation]{General bright-dark soliton solution\\ to semi-discrete vector NLS equation}\label{section2}

Let us consider a general soliton solution consisting of $m$-bright solitons and ($M-m)$-dark solitons to the semi-discrete vector NLS equation~(\ref{MCNLS_sd}). To this end, we introduce the following dependent variable transformations
\begin{gather} \label{dependentvar_trf}
q_{n}^{(j)}=\mathrm{i}^{n}\frac{g_{n}^{(j)}}{f_{n}} ,\qquad
q_{n}^{(m+l)}=\rho _{l}(\mathrm{i}a_{l})^{n}\frac{h_{n}^{(l)}}{f_{n}}e^{\omega _{l}t} ,
\end{gather}
where $j=1,\dots ,m$, $\omega_{l}=s(a_{l}-\bar{a}_{l})$, $|a_{l}|=1$ with $\bar{a}_{l}$ representing the complex conjugate of $a_{l}$, $l=1,\dots,M-m$. Here, $f_n$ is a real-valued function, whereas, $g_n$ and $h_n$ are complex-valued functions. The transformations convert equation~(\ref{MCNLS_sd}) into a set of bilinear equations as follows
\begin{gather}
 D_{t}g_{n}^{(j)}\cdot f_{n}=s\big(g_{n+1}^{(j)}f_{n-1}-g_{n-1}^{(j)}f_{n+1}\big), \qquad j=1,\dots ,m ,\nonumber\\
 ( D_{t}+\omega _{l} ) h_{n}^{(l)}\cdot f_{n}=s\big(a_{l}h_{n+1}^{(l)}f_{n-1}-\bar{a}_{l}h_{n-1}^{(l)}f_{n+1}\big), \qquad l=1,\dots ,M-m , \nonumber\\
sf_{n+1}f_{n-1}-f_{n}^{2}=\sum_{j=1}^{m}\sigma _{j}\big|g_{n}^{(j)}\big|^{2}+\sum_{l=1}^{M-m}\sigma _{l+m}|\rho _{l}|^{2}\big|h_{n}^{(l)}\big|^{2} .
 \label{MCNLS-bilinear_sd}
\end{gather}
Here $s=1+\sum\limits_{l=1}^{M-m}\sigma _{l+m}|\rho _{l}|^{2}$.

In what follows, we give a Pfaf\/f\/ian-type solution to the above bilinear equations.
\begin{Theorem}A set of bilinear equations \eqref{MCNLS-bilinear_sd} admit the following solutions in the form of Pfaffians
\begin{gather*}%\label{cHirota_Nsoliton2}
f_{n} =\operatorname{Pf} (a_{1},\dots ,a_{2N},b_{1},\dots ,b_{2N}) ,\\
g_{n}^{(j)} =\operatorname{Pf} (d_{0},\beta _{j},a_{1},\dots,a_{2N},b_{1},\dots ,b_{2N}) , \qquad
 h_{n}^{(l)} =\operatorname{Pf} \big(c_{1}^{(l)},\dots ,c_{2N}^{(l)},b_{1},\dots,b_{2N}\big) ,
\end{gather*}
with the elements of the Pfaffians defined as follows
\begin{gather*}%\label{CCHirota_pf1}
\operatorname{Pf} (a_{j},a_{k})_{n}=\frac{p_{j}-p_{k}}{p_{j}p_{k}-1}\varphi
_{j}(n)\varphi _{k}(n) , \qquad \operatorname{Pf} (d_{0},b_{j})=\operatorname{Pf} (d_{0},\beta _{l})=0 ,\\
\operatorname{Pf} \big(c_{j}^{(l)},c_{k}^{(l)}\big)_{n}=\frac{p_{j}-p_{k}}{p_{j}p_{k}-1}
\frac{p_{j}-a_{l}}{1-a_{l}p_{j}}\frac{p_{k}-a_{l}}{1-a_{l}p_{k}}\varphi
_{j}(n)\varphi _{k}(n) , %\label{CCHirota_pf11}
\\
\operatorname{Pf} (a_{j},b_{k})=\delta _{jk} ,\qquad \operatorname{Pf}
\big(c_{j}^{(l)},b_{k}\big)=\delta _{jk},\qquad \operatorname{Pf} (d_{l},a_{k})_{n}=p_{k}^{l}\varphi _{k}(n) , %\label{CCHirota_pf3}
\\
\operatorname{Pf} (b_{j},\beta _{l})=
\begin{cases}
0, & 1 \le j \le N, \\
\alpha^{(l)}_{j-N}, & N+1 \le j \le 2N ,
\end{cases} \qquad \operatorname{Pf} (a_{j},\beta _{l})=0 , %\label{CCHirota_pf4}
\\
\operatorname{Pf} (b_{j},b_{k}) =
\begin{cases}
b_{jk}, & 1 \le j \le N, \ N+1 \le k \le 2N, \\
0, & {\mathrm{otherwise}},
\end{cases} %\label{CCHirota_pf12}
\end{gather*}
with
\begin{gather*}
b_{jk}= \frac{\sum\limits_{l=1}^{m}\alpha _{j}^{(l)}\sigma _{l} \overline{\alpha _{k-N}^{(l)}}}{(p_{j}-p_{k})(p_{j}p_{k}-1)\left( \frac{s}{p_{j}p_{k}}%
-\sum\limits_{l=1}^{M-m}\frac{\sigma _{l+m}|\rho _{l}|^{2}(a_{l}-\bar{a}_{l})^{2}}{(1-a_{l}p_{j})(1-a_{l}p_{k})(1-\bar{a}_{l}p_{j})(1-\bar{a}_{l}p_{k})}\right)
} ,
\end{gather*}
and $\varphi _{j}(n)=p_{j}^{n}e^{\eta _{j}}$, $\eta _{j}=s\big(p_{j}-p_{j}^{-1}\big)t$ $+\eta _{j,0}$ which satisfying $p_{j+N}=\bar{p}_{j}$, $\eta _{j+N,0}=\bar{\eta }_{j,0}$.
\end{Theorem}

\begin{proof}It can be shown that
\begin{gather*}
\frac{{\rm d}}{{\rm d}t}f_{n}=s\operatorname{Pf} (d_{-1},d_{1},\dots )_{n} , \qquad \text{where} \quad \operatorname{Pf} (d_{-1},d_{1})\equiv 0 , \qquad \operatorname{Pf} (d_{\pm 1},b_{j})\equiv 0 ,\\
f_{n+1} = \operatorname{Pf} (d_0,d_{1}, \dots)_n , \qquad \text{where} \quad \operatorname{Pf} (d_{0},d_{1})\equiv 1 ,
\\
f_{n-1} = \operatorname{Pf} (d_0,d_{-1}, \dots)_n , \qquad \text{where} \quad \operatorname{Pf} (d_{0},d_{-1})\equiv 1,
\end{gather*}
and
\begin{gather*}
\frac{{\rm d}}{{\rm d}t}g_{n}^{(j)}=s\operatorname{Pf} (d_{0},d_{-1},d_{1},\beta _{j},\dots)_{n} ,\\
g_{n+1}^{(j)}=\operatorname{Pf} (d_{1},\beta _{j},\dots )_{n} ,\qquad g_{n-1}^{(j)}=\operatorname{Pf} (d_{-1},\beta _{j},\dots )_{n} ,
\end{gather*}
where $\operatorname{Pf} (d_0,d_1, a_1, \dots, a_{2N}, b_1, \dots , b_{2N})$ is abbreviated by $\operatorname{Pf} (d_0, d_1, \dots)$, so as other similar Pfaf\/f\/ians. Thus, an algebraic identity of Pfaf\/f\/ian
\begin{gather*}
\operatorname{Pf} (d_{0},d_{-1},d_{1},\beta_j,\dots )_{n}\operatorname{Pf} (\dots)_{n}=\operatorname{Pf} (d_{0},d_{-1},\dots )_{n}\operatorname{Pf} (d_{1},\beta_j,\dots)_{n} \\
\qquad {}-\operatorname{Pf} (d_{0},d_{1},\dots )_{n}\operatorname{Pf} (d_{-1},\beta_j ,\dots )_{n}+\operatorname{Pf} (d_{0},\beta_j,\dots )_{n}\operatorname{Pf} (d_{-1},d_{1},\dots )_{n} ,
\end{gather*}
together with above Pfaf\/f\/ian relations gives
\begin{gather*}
\left( \frac{{\rm d}}{{\rm d}t}g^{(j)}_{n}\right) \times f_{n}=sg^{(j)}_{n+1}\times
f_{n-1}-sg^{(j)}_{n-1}\times f_{n+1}+g^{(j)}_{n}\times \left( \frac{{\rm d}}{{\rm d}t}f_{n}\right) ,
\end{gather*}
which is exactly the f\/irst bilinear equation. Next we prove the second bilinear equation. It can also be shown that
\begin{gather*}
h_{n}^{(l)}=\operatorname{Pf} \big(d_{0},\bar{d}^{(l)}_{0},\dots \big)_{n} ,\qquad
h_{n+1}^{(l)}=\bar{a}_{l}\operatorname{Pf} \big(d_{1},\bar{d}^{(l)}_{0},\dots \big)_{n} ,\qquad
h_{n-1}^{(l)}={a}_{l}\operatorname{Pf} \big(d_{-1},\bar{d}^{(l)}_{0},\dots \big)_{n} ,\!\\
\left( \frac{{\rm d}}{{\rm d}t}+s(a_{l}-\bar{a}_{l})\right) h_{n}^{(l)}=s\operatorname{Pf} \big(d_{0},d_{-1},d_{1},\bar{d}^{(l)}_{0},\dots \big)_{n} ,
\end{gather*}
where
\begin{gather*}
\operatorname{Pf} \big(\bar{d}^{(l)}_{0},a_{j}\big)=p_{j}^{n}\left( \frac{p_{j}-a_{l}}{1-a_{l}p_{j}}\right) e^{\eta _{j}} ,\qquad \operatorname{Pf} \big(\bar{d}^{(l)}_{0},b_{j}\big)=0 ,\qquad \operatorname{Pf} \big(d_{0},\bar{d}^{(l)}_{0}\big)=1,\\
\operatorname{Pf} \big(d_{-1,}\bar{d}^{(l)}_{0}\big)=\bar{a}_{l},\qquad \operatorname{Pf} \big(d_{1},\bar{d}^{(l)}_{0}\big)=a_{l} .
\end{gather*}
Therefore, an algebraic identity of Pfaf\/f\/ian
\begin{gather*}
\operatorname{Pf} \big(d_{0},d_{-1},d_{1},\bar{d}^{(l)}_{0},\dots \big)_{n}\operatorname{Pf} (\dots )_{n}=\operatorname{Pf} (d_{0},d_{-1},\dots )_{n}\operatorname{Pf} \big(d_{1},\bar{d}^{(l)}_{0},\dots \big)_{n} \\
\qquad {}-\operatorname{Pf} (d_{0},d_{1},\dots )_{n}\operatorname{Pf} (d_{-1},\bar{d_{0}},\dots )_{n}+\operatorname{Pf} \big(d_{0},\bar{d}^{(l)}_{0},\dots \big)_{n}\operatorname{Pf}(d_{-1},d_{1},\dots )_{n} ,
\end{gather*}
together with above Pfaf\/f\/ian relations gives
\begin{gather*}
\left( \frac{{\rm d}}{{\rm d}t}+s(a_{l}-\bar{a}_{l})\right) h_{n}^{(l)}\times f_{n}=s\big(a_{l}h_{n+1}^{(l)} f_{n-1}-\bar{a}_{l}h_{n-1}^{(l)}
f_{n+1}\big)+h_{n}^{(l)} \left( \frac{{\rm d}}{{\rm d}t}f_{n}\right) ,
\end{gather*}
which is nothing but the second bilinear equation. Now let us proceed to the proof of the third bilinear equation. To this end, we need to def\/ine
\begin{gather*}
\operatorname{Pf} (a_{j},\bar{\beta}_{l})=0 ,\qquad \operatorname{Pf} (b_{j},\bar{\beta}_{l}) =
\begin{cases}
0, & 1 \le j \le N ,\\
\overline{\alpha^{(l)}_{j-N}}, & N+1 \le j \le 2N ,
\end{cases}\\
\operatorname{Pf} \big(d_{0},{\overline{d_{0}'}}^{(l)}\big)=1 , \qquad
\operatorname{Pf} \big(\bar{c}_{j}^{(l)},\bar{c}_{k}^{(l)}\big)_{n} =\frac{p_{j}-p_{k}}{p_{j}p_{k}-1}\frac{1-a_{l}p_{j}}{p_{j}-a_{l}}\frac{1-a_{l}p_{k}
}{p_{k}-a_{l}}\varphi _{j}(n)\varphi _{k}(n) ,\\
\operatorname{Pf} \big({\overline{d_{0}'}}^{(l)},a_{j}\big)=p_{j}^{n}\left( \frac{1-a_{l}p_{j}}{p_{j}-a_{l}}\right) e^{\eta _{j}} , \qquad
\operatorname{Pf} \big(\bar{c}_{j}^{(l)},b_{k}\big) =\delta_{jk} , \qquad
\operatorname{Pf} \big({\overline{d_{0}'}}^{(l)},b_{j}\big)=0 .
\end{gather*}
Then from the fact
\begin{gather*}
\overline{\operatorname{Pf} ({a}_{j},{a}_{k})}=\operatorname{Pf} (a_{j^{\prime}},a_{k^{\prime }}) ,\qquad \overline{\operatorname{Pf} ({b}_{j},{b}_{k})}=\operatorname{Pf}(b_{j^{\prime }},b_{k^{\prime }}) ,
\end{gather*}%
where $j^{\prime }=j+N$ mod$(2N)$, $k^{\prime }=k+N$ mod$(2N)$, we obtain
\begin{gather*}
\bar{f}_{n} =f_{n} ,\qquad \bar{g}_{n}^{(j)}=\operatorname{Pf} (d_{0},\bar{\beta}_{j},a_{1},\dots ,a_{2N},b_{1},\dots ,b_{2N})_{n} , \\
 \bar{h}_{n}^{(l)}=\operatorname{Pf} \big(\bar{c}_{1}^{(l)},\dots ,\bar{c}_{2N}^{(l)},b_{1},\dots ,b_{2N}\big) =\operatorname{Pf} \big(d_{0},{\overline{d_{0}'}}^{(l)},\dots \big)_{n} .
\end{gather*}
Since
\begin{gather*}
f_{n+1} = \operatorname{Pf} (d_0,d_{1}, \dots)_n ,\qquad f_{n-1}=\operatorname{Pf} (d_{0},d_{-1},\dots )_{n} ,
\end{gather*}
we then have
\begin{gather*}
f_{n+1}=f_{n}+\sum_{j=1}^{2N}(-1)^{j-1}\operatorname{Pf} (d_{1},a_{j})\operatorname{Pf} (d_{0},\dots ,\hat{a}_{j},\dots )_{n} ,\\
f_{n-1} = f_{n}+ \sum_{j=1}^{2N} (-1)^{j-1} \operatorname{Pf} (d_{-1}, a_j)\operatorname{Pf} (d_0, \dots, \hat{a}_j, \dots)_n ,
\end{gather*}
Then we can show
\begin{gather*}
f_{n+1}f_{n-1}-f_{n}f_{n} =-\sum_{j<k}(-1)^{j+k}\left(p_{j}+\frac{1}{p_{j}}-p_{k}-\frac{1}{p_{k}}\right)\operatorname{Pf}(b_{j,}b_{k}) \\
\hphantom{f_{n+1}f_{n-1}-f_{n}f_{n} =}{} \times \operatorname{Pf} (d_{0},\dots ,\hat{b}_{j},\dots )\operatorname{Pf}
(d_{0},\dots ,\hat{b}_{k},\dots ) .
\end{gather*}
On the other hand, since
\begin{gather*}
h_{n}^{(l)} =\operatorname{Pf} \big(d_{0},\overline{d_{0}}^{(l)},\dots \big)_{n} , \qquad \bar{h}_{n}^{(l)} =\operatorname{Pf} \big(d_{0},{\overline{d_{0}'}}^{(l)},\dots \big)_{n} ,
\end{gather*}
we have
\begin{gather*}
h_{n}^{(l)}=f_{n}+\sum_{j=1}^{2N}(-1)^{j-1}\operatorname{Pf} (d_{0},a_{j})_{n}\left(
\frac{p_{j}-a_{l}}{1-a_{l}p_{j}}\right)\operatorname{Pf} (d_{0},\dots ,\hat{a}_{j},\dots )_{n} ,\\
\bar{h}_{n}^{(l)}=f_{n}+\sum_{j=1}^{2N}(-1)^{j-1}\operatorname{Pf}
(d_{0},a_{j})_{n}\left(\frac{1-a_{l}p_{j}}{p_{j}-a_{l}}\right)\operatorname{Pf} (d_{0},\dots ,\hat{a}_{j},\dots )_{n} .
\end{gather*}
Similarly, we can show
\begin{gather*}
\big|h_{n}^{(l)}\big|^{2}-f_{n}f_{n} =-\sum_{j<k}(-1)^{j+k}\left( \frac{p_{j}-a_{l}}{1-a_{l}p_{j}}+\frac{1-a_{l}p_{j}}{p_{j}-a_{l}}-\frac{p_{k}-a_{l}}{1-a_{l}p_{k}}-\frac{1-a_{l}p_{k}}{p_{k}-a_{l}}\right) \\
\hphantom{\big|h_{n}^{(l)}\big|^{2}-f_{n}f_{n} =}{} \times \operatorname{Pf} (b_{j,}b_{k})\operatorname{Pf} (d_{0},\dots ,\hat{b}_{j},\dots )\operatorname{Pf} (d_{0},\dots ,\hat{b}_{k},\dots ) .
\end{gather*}
Finally, we have
\begin{gather*}
s(f_{n+1}f_{n-1}-f_{n}f_{n})-\sum_{l=1}^{M-m}\sigma _{l+m}|\rho_{l}|^{2}\big( |h_{n}^{(l)}|^{2}-f_{n}f_{n}\big) \\
\qquad{} =-\sum_{j<k}(-1)^{j+k}\left\{s\left(p_{j}+\frac{1}{p_{j}}-p_{k}-\frac{1}{p_{k}}\right) \right.\\
\left. \qquad\quad{}-\sum_{l=1}^{M-m}\sigma _{l+m}|\rho _{l}|^{2}\left( \frac{p_{j}-a_{l}}{1-a_{l}p_{j}}+\frac{1-a_{l}p_{j}}{p_{j}-a_{l}} -\frac{p_{k}-a_{l}}{1-a_{l}p_{k}}-\frac{1-a_{l}p_{k}}{p_{k}-a_{l}}\right)\right\} \\
\qquad\quad{} \times \operatorname{Pf} (b_{j,}b_{k})\operatorname{Pf} (d_{0},\dots ,\hat{b}_{j},\dots )\operatorname{Pf} (d_{0},\dots ,\hat{b}_{k},\dots ) \\
\qquad{} =\sum_{j=1}^{N}\sum_{k=N+1}^{2N}(-1)^{j+k}\sum_{l=1}^{m}\alpha_{j}^{(l)}\sigma _{l} \overline{\alpha _{k-N}^{(l)}} \operatorname{Pf} (d_{0},\dots ,\hat{b}_{j},\dots )\operatorname{Pf} (d_{0},\dots ,\hat{b}_{k},\dots ) \\
\qquad {} =\sum_{l=1}^{m}\sigma _{l}\sum_{j=1}^{N}\sum_{k=N+1}^{2N}(-1)^{j+k}\operatorname{Pf} (b_{j},\beta _{l})\operatorname{Pf} (b_{k},\bar{\beta }_{l})\operatorname{Pf} (d_{0},\dots ,\hat{b}_{j},\dots )\operatorname{Pf} (d_{0},\dots ,\hat{b}_{k},\dots ) \\
\qquad{} =\sum_{l=1}^{m}\sigma _{l}|g_{n}^{(l)}|^{2} .
\end{gather*}
The third bilinear equation is proved.
\end{proof}

The above Pfaf\/f\/ian solutions, with dependent variable transformations (\ref{dependentvar_trf}), give general $N$-bright-dark soliton solutions to the semi-discrete vector NLS equation (\ref{MCNLS_sd}).
\section[One- and two-soliton solutions for the two- and three-coupled discrete NLS equation]{One- and two-soliton solutions for the two-\\ and three-coupled discrete NLS equation}\label{section3}
\subsection{Two-component semi-discrete NLS equation}
In this subsection, we provide and illustrate one- and two-soliton for two-component semi-discrete NLS equation (\ref{sdCNLS}) explicitly.

{\bf{One-soliton solution}.} the tau-functions for one-soliton solution ($N=1$) are
\begin{gather*}
f_{n}=-1-c_{1\bar{1}}(p_{1}\bar{p}_{1})^{n}e^{\eta _{1}+\bar{\eta}_{1}} ,\\
g_{n}^{(1)}=-\alpha _{1}^{(1)}p_{1}{}^{n}e^{\eta _{1}} ,\qquad
h_{n}^{(1)}=-1-d_{1\bar{1}}^{(1)}(p_{1}\bar{p}_{1})^{n}e^{\eta _{1}+\bar{\eta}_{1}} ,
\end{gather*}
where
\begin{gather*}
c_{1\bar{1}}=\frac{\bar{\alpha}_{1}^{(1)}\sigma _{1}\alpha _{1}^{(1)}}{({p}_{1}\bar{p}_{1}-1)^{2}\left(\frac{s}{|{p}_{1}|^2}+\frac{{\sigma _{2}|\rho
_{1}|^{2}}(a_{1}-\bar{a}_{1})^{2}}{|1-a_{1}p_{1}|^{2}|1-\bar{a}_{1}p_{1}|^{2}}\right)} ,\qquad
d_{1\bar{1}}^{(1)}=-\frac{(p_{1}-a_{1})(\bar{p}_{1}-a_{1})}{(1-a_{1}p_{1})(1-a_{1}\bar{p}_{1})}c_{1\bar{1}}.
\end{gather*}

The above tau functions lead to the one-soliton solution as follows
\begin{gather*}
q_{n}^{(1)}=\frac{\alpha _{1}^{(1)}}{2\sqrt{c_{1\bar{1}}}}e^{\mathrm{i}\xi_{1I}}\operatorname{sech}( \xi _{1R}+\theta _{0}) ,\\
q_{n}^{(2)}=\frac{1}{2}\rho _{1}e^{\mathrm{i}\zeta _{1}}\big( 1+e^{2\mathrm{i}\phi _{1}}+(e^{2\mathrm{i}\phi _{1}}-1)\tanh ( \xi _{1R}+\theta _{0}) \big) ,
\end{gather*}
where $\xi _{1}=\xi _{1R}+\mathrm{i}\xi _{1I}=n\ln (\mathrm{i}p_{1})+s\big(p_{1}-p_{1}^{-1}\big)t$, $\zeta _{1}=n\varphi _{1}+n\pi /2+\omega _{1}t$, $e^{\mathrm{i}\varphi _{1}}=a_{1}$, $e^{2\theta _{0}}=c_{1\bar{1}}$, $e^{2\mathrm{i}\phi _{1}}= -(p_{1}-a_{1})(\bar{p}_{1}-a_{1})/((1-a_{1}p_{1})(1-a_{1}\bar{p}_{1}))$. Therefore, the amplitude of bright soliton for $q^{(1)}$ are $\frac{1}{2}\big|\alpha _{1}^{(1)}\big|/\sqrt{a_{1\bar{1}}}$. The dark soliton $q^{(2)}$ approaches $|\rho _{1}|$ as $x\rightarrow \pm \infty $. In addition, the intensity of the dark soliton is $|\rho _{1}|\cos \phi _{1}$.

An example of one-bright-dark soliton is illustrated in Fig.~\ref{fig1}a $p_1= 1.0+0.8 \mathrm{i}$, $\rho_1=5.0$, $\alpha _{1}^{(1)}=1.0+\mathrm{i}$, $a_1=0.8+0.6\mathrm{i}$. for focusing-focusing case ($\sigma_1=1.0$, $\sigma_2=1.0$). Fig.~\ref{fig1}b shows the bright-dark soliton for the focusing-defocusing case ($\sigma_1=1.0$, $\sigma_2=-1.0$). It is interesting to note the dark soliton corresponding to defocusing component becomes an anti-dark one.

\begin{figure}[t]\centering
\includegraphics[scale=0.37]{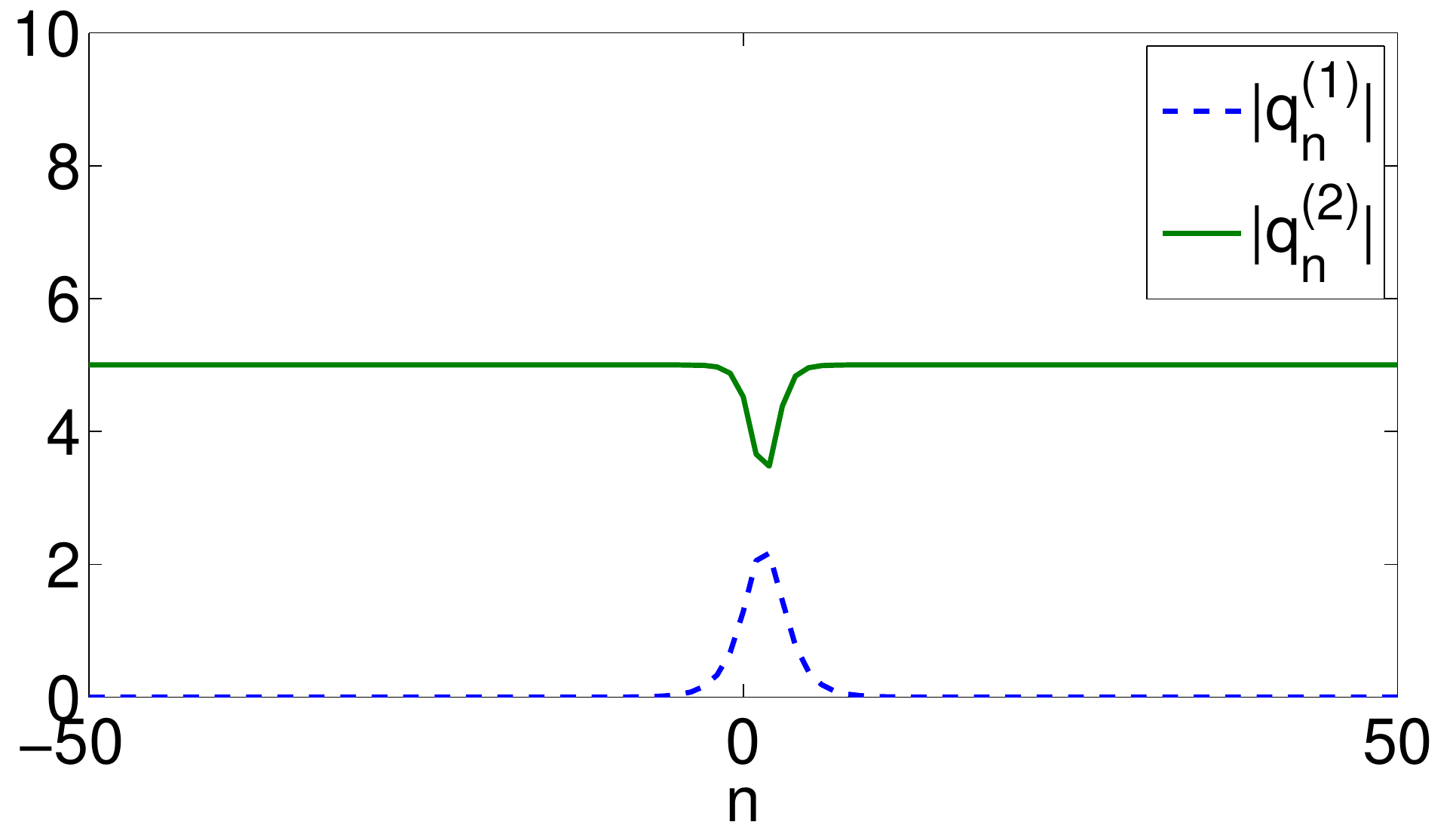}\qquad \includegraphics[scale=0.37]{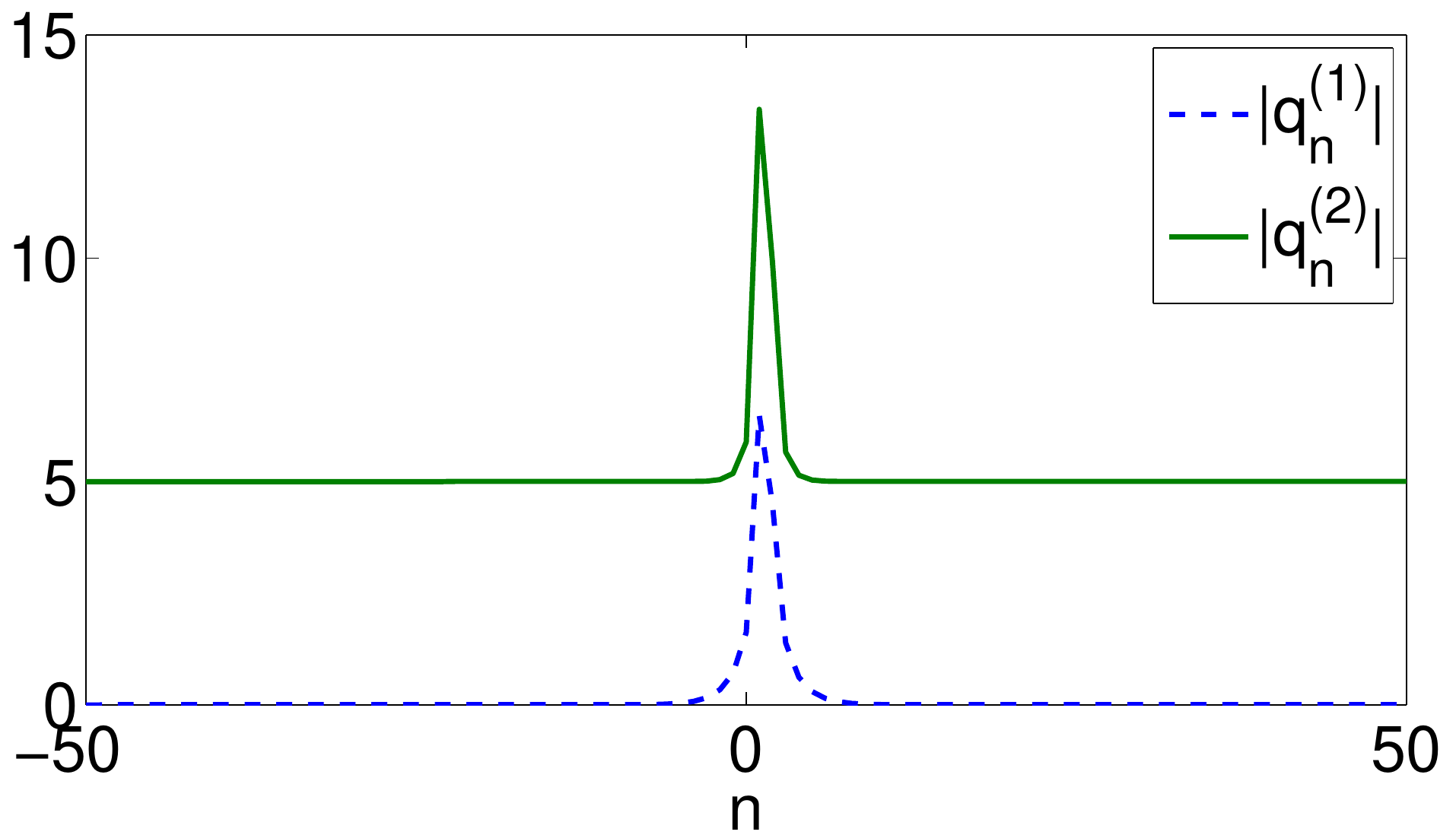}
\caption{One-bright-dark soliton soliton solution to a two-coupled semi-discrete NLS equation: (a)~focusing-focusing case ($\sigma_1=1.0$, $\sigma_2=1.0$); (b)~focusing-defocusing case ($\sigma_1=1.0$, $\sigma_2=-1.0$).}\label{fig1}
\end{figure}

{\bf{Two-soliton solution}.} The tau functions for two-soliton are of the following form
\begin{gather*}
f_{n}=1+c_{1\bar{1}}E_{1}\bar{E}_{1}+c_{2\bar{1}}E_{2}\bar{E}_{1}+c_{1\bar{2}}E_{1}\bar{E}_{2}+c_{2\bar{2}}E_{2}\bar{E}_{2}+c_{12\bar{1}\bar{2}}E_{1}E_{2}
\bar{E}_{1}\bar{E}_{2} , \\ %\label{twosoliton_1bt1dk1}
g_{n}=\alpha _{1}^{(1)}E_{1}+\alpha _{2}^{(1)}E_{2}+c_{12\bar{1}}^{(j)}E_{1}E_{2}\bar{E}_{1}+c_{12\bar{2}}^{(j)}E_{1}E_{2}\bar{E}_{2} ,\\
%\label{twosoliton_1bt1dk2}
h_{n}=1+d_{1\bar{1}}^{(1)}E_{1}\bar{E}_{1}+d_{2\bar{1}}^{(1)}E_{2}\bar{E}_{1}+d_{1\bar{2}}^{(1)} E_{1}\bar{E}_{2}+d_{2\bar{2}}^{(1)}E_{2}\bar{E}_{2}+d_{12\bar{1}\bar{2}}^{(1)}E_{1}E_{2}\bar{E}_{1}\bar{E}_{2} , %\label{twosoliton_1bt1dk3}
\end{gather*}
where $E_{j}=p_{j}{}^{n}e^{\eta _{j}}$,
\begin{gather*}
c_{i\bar{j}}=\frac{\bar{\alpha}_{i}^{(1)}\sigma _{1}\alpha _{j}^{(1)}}{({p}_{i}\bar{p}_{j}-1)^{2}\left(\frac{s}{{p}_{i}\bar{p}_{j}}+\frac{{\sigma_{2}
|\rho_{1}|^{2}}(a_{1}-\bar{a}_{1})^{2}}{|1-a_{1}p_{i}|^{2}|1-\bar{a}_{1}p_{j}|^{2}}\right)} ,\qquad
d_{i\bar{j}}^{(1)}=-\frac{(p_{i}-a_{1})(\bar{p}_{j}-a_{1})}{(1-a_{1}p_{i})(1-a_{1}\bar{p}_{j})}c_{i\bar{j}} ,\\
c_{12\bar{1}\bar{2}}=|p_{2}-p_{1}|^{2}\left( \frac{c_{1\bar{1}}c_{2\bar{2}}}{(p_{1}+\bar{p}_{2})\left( p_{2}+\bar{p}_{1}\right) }-\frac{c_{1\bar{2}}c_{2\bar{1}}}{(p_{1}+\bar{p}_{1}) ( p_{2}+\bar{p}_{2} ) }\right) ,\\
c_{12\bar{j}}= ( p_{2}-p_{1} ) \left( \frac{\alpha _{2}^{(1)}c_{1\bar{j}}}{p_{2}+\bar{p}_{j}}-\frac{\alpha _{1}^{(1)}c_{2\bar{j}}}{p_{1}+\bar{p}_{j}}\right),\\
d_{12\bar{1}\bar{2}} =\frac{(p_{1}-a_{1})(\bar{p}_{1}-a_{1})(p_{2}-a_{2})(\bar{p}_{1}-a_{2})}{(1-a_{1}p_{1}) (1-a_{1}\bar{p}_{1})(1-a_{1}p_{2})(1-a_{1}\bar{p}_{2})}c_{12\bar{1}\bar{2}}.
\end{gather*}
A two bright-dark soliton solution is shown in Fig.~\ref{fig2} before and after the collision for parameters $\sigma_1=1.0$, $\sigma_2=-1$, $p_1= 1.0+0.5 \mathrm{i}$, $\rho=2.0$, $\alpha _{1}^{(1)}=1.0+0.8\mathrm{i}$, $\alpha _{2}^{(1)}=1.0+0.5\mathrm{i}$, $a_1=0.6+0.8\mathrm{i}$. It can be seen that the collision is elastic which is the same as for the continuous two-component NLS equation.

\begin{figure}[t]\centering
\includegraphics[scale=0.37]{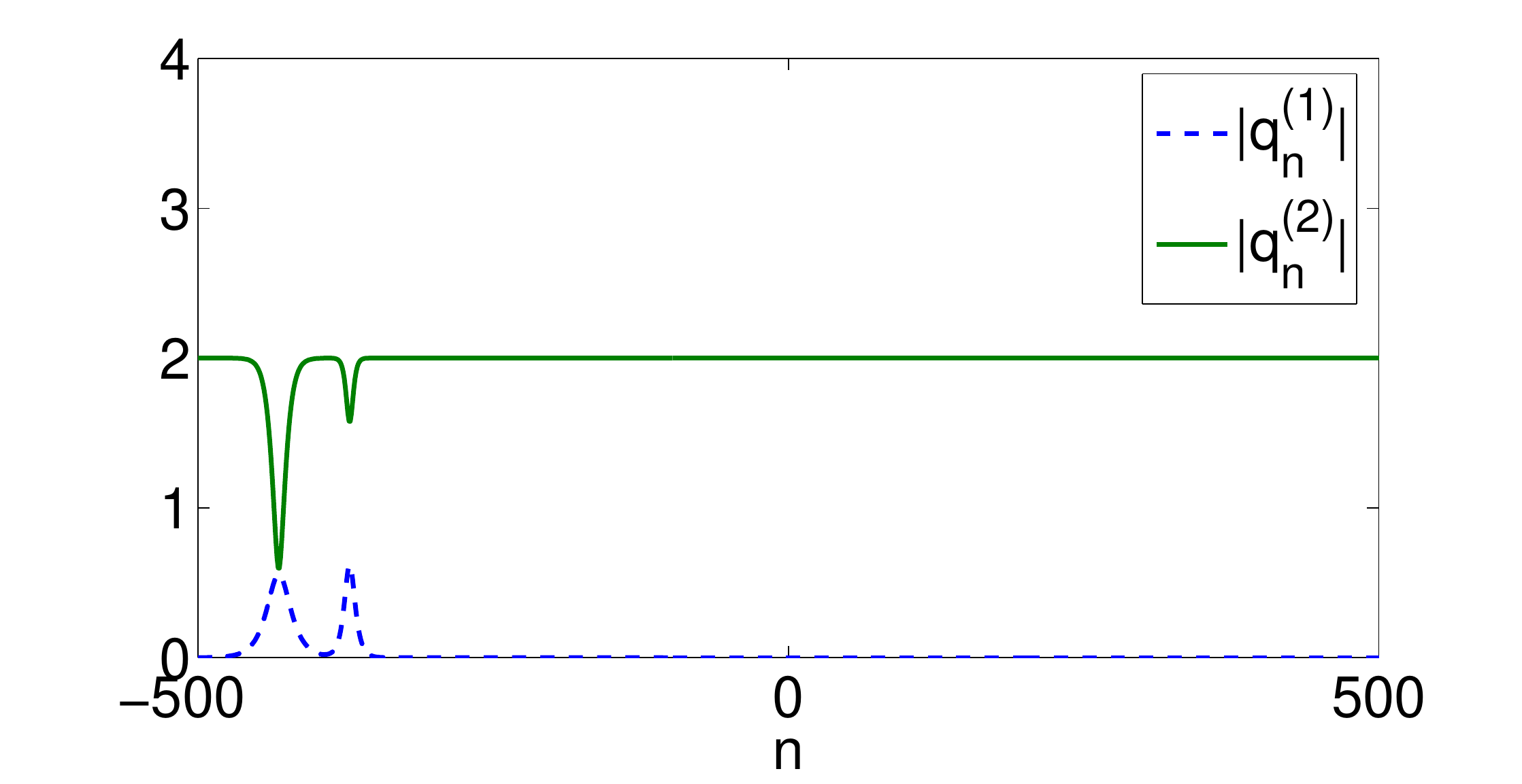}\qquad\includegraphics[scale=0.37]{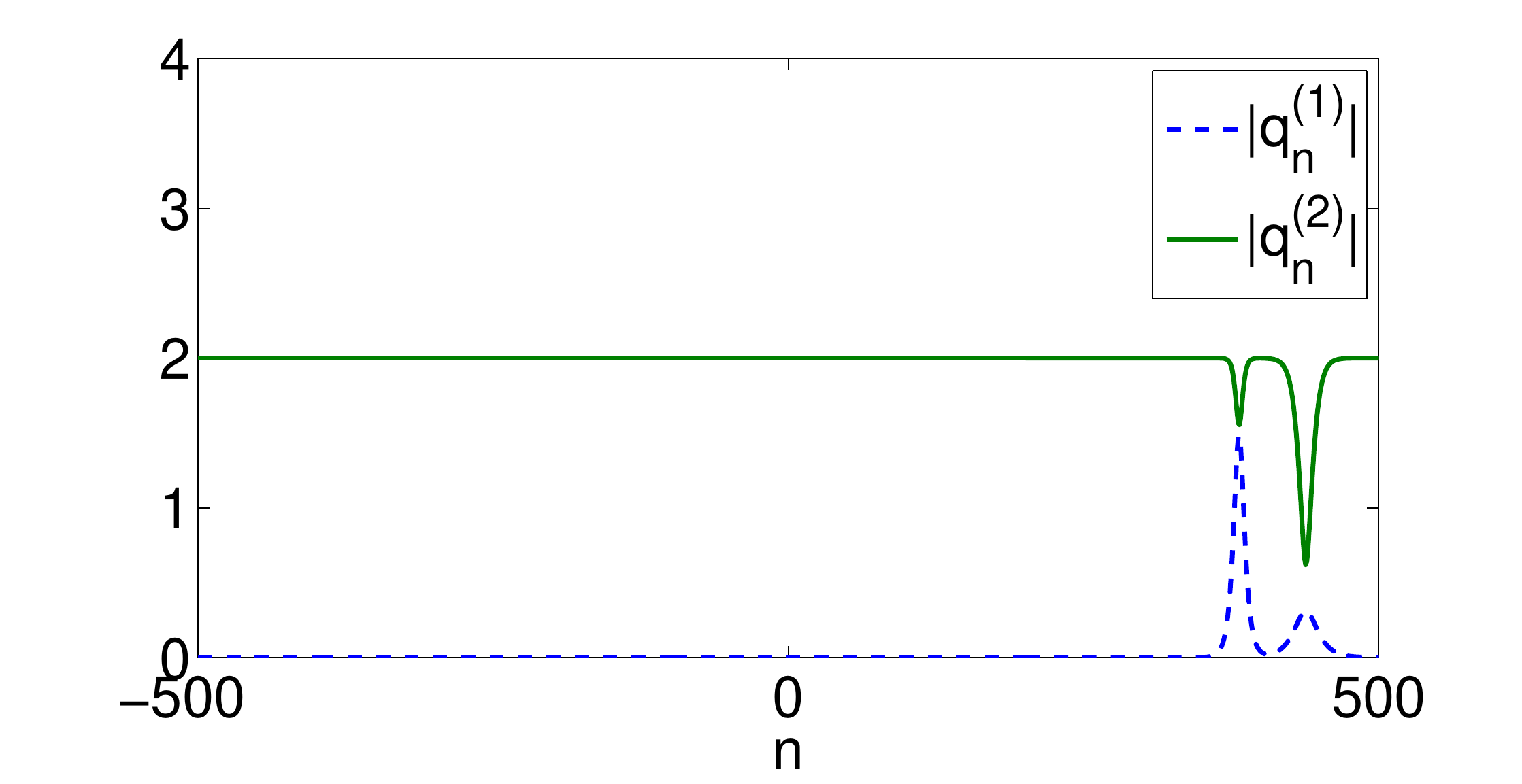}

\caption{A two-soliton solution of mixed type for a two-coupled semi-discrete NLS equation: (a)~before the collision $t=-80$; (b)~after the collision $t=80$.}
\label{fig2}
\end{figure}

\subsection[Bright-dark soliton solution for three-component semi-discrete NLS equation]{Bright-dark soliton solution for three-component\\ semi-discrete NLS equation}

In this subsection, we will give the mixed-type soliton solution to the following three-component semi-discrete NLS equation
\begin{gather*}
\mathrm{i}q_{n,t}^{(1)}=\left(1+\sum_{k=1}^{3}\sigma_{k}\big|q_{n}^{(k)}\big|^{2}\right)\big(q_{n+1}^{(1)}+q_{n-1}^{(1)}\big) , \\
\mathrm{i}q_{n,t}^{(2)}=\left(1+\sum_{k=1}^{3}\sigma_{k}\big|q_{n}^{(k)}\big|^{2}\right)\big(q_{n+1}^{(2)}+q_{n-1}^{(2)}\big) , \\
\mathrm{i}q_{n,t}^{(3)}=\left(1+\sum_{k=1}^{3}\sigma_{k}\big|q_{n}^{(k)}\big|^{2}\right)\left(q_{n+1}^{(3)}+q_{n-1}^{(3)}\right) .
%\label{sd_3NLS}
\end{gather*}

{\bf{Two-bright-one-dark soliton solution}.} The tau-functions for one-soliton solution \linebreak ($N=1$) are
\begin{gather*}
f_{n}=-1-c_{1\bar{1}}(p_{1}\bar{p}_{1})^{n}e^{\eta _{1}+\bar{\eta}_{1}} ,\qquad g_{n}^{(j)}=-\alpha _{1}^{(j)}p_{1}{}^{n}e^{\eta _{1}} ,
\qquad h_{n}^{(1)}=-1-d_{1\bar{1}}^{(1)}(p_{1}\bar{p}_{1})^{n}e^{\eta _{1}+\bar{\eta}_{1}} ,
\end{gather*}
where
\begin{gather*}
c_{1\bar{1}}=\frac{\sum\limits_{k=1}^{2}\bar{\alpha}_{1}^{(k)}\sigma _{k}\alpha
_{1}^{(k)}}{({p}_{1}\bar{p}_{1}-1)^{2}\left(\frac{s}{|{p}_{1}|^2}+\frac{{\sigma _{3}|\rho
_{1}|^{2}}(a_{1}-\bar{a}_{1})^{2}}{|1-a_{1}p_{1}|^{2}|1-\bar{a}_{1}p_{1}|^{2}}\right) } ,\qquad
d_{1\bar{1}}^{(1)}=-\frac{(p_{1}-a_{1})(\bar{p}_{1}-a_{1})}{(1-a_{1}p_{1})(1-a_{1}\bar{p}_{1})}c_{1\bar{1}} ,
\end{gather*}
where $\omega =s(a_{1}-\bar{a}_{1})$, $|a_{1}|=1$, $s=1+\sigma _{3}|\rho _{1}|^{2}$ with $\bar{a}_{1}$ representing the complex conjugate of~$a_{1}$. The above tau functions lead to the one-soliton as follows
\begin{gather*}
q_{n}^{(j)}=\frac{\alpha _{1}^{(j)}}{2\sqrt{c_{1\bar{1}}}}e^{\mathrm{i}\xi_{1I}}\operatorname{sech} ( \xi _{1R}+\theta _{0} ) ,\qquad j=1,2\\
q_{n}^{(3)}=\frac{1}{2}\rho _{1}e^{\mathrm{i}\zeta _{1}}\big( 1+e^{2\mathrm{i}\phi _{1}}+(e^{2\mathrm{i}\phi _{1}}-1)\tanh ( \xi _{1R}+\theta _{0}) \big) ,
\end{gather*}
where $\xi _{1}=\xi _{1R}+\mathrm{i}\xi _{1I}=n\ln (\mathrm{i}p_{1})+s\big(p_{1}-p_{1}^{-1}\big)t$, $\zeta _{1}=n\varphi _{1}+n\pi /2+\omega _{1}t$, $e^{\mathrm{i}\varphi _{1}}=a_{1}$, $e^{2\theta _{0}}=c_{1\bar{1}}$, $e^{2\mathrm{i}\phi _{1}}= -(p_{1}-a_{1})(\bar{p}_{1}-a_{1})/ ((1-a_{1}p_{1})(1-a_{1}\bar{p}_{1}))$. Therefore, the amplitude of bright soliton for $q^{(j)}$ are $\frac{1}{2}\big|\alpha _{1}^{(j)}\big|/\sqrt{a_{1\bar{1}}}$. The dark soliton $q^{(3)}$ approaches $|\rho _{1}|$ as $x\rightarrow \pm \infty $. In addition, the intensity of the dark soliton is $|\rho _{1}|\cos \phi _{1}$.

The tau functions for two-soliton are of the following form
\begin{gather*}
f_{n}=1+c_{1\bar{1}}E_{1}\bar{E}_{1}+c_{2\bar{1}}E_{2}\bar{E}_{1}+c_{1\bar{2}}E_{1}\bar{E}_{2}+c_{2\bar{2}}E_{2}\bar{E}_{2}+c_{12\bar{1}\bar{2}}E_{1}E_{2}
\bar{E}_{1}\bar{E}_{2} , \\
g_{n}^{(j)}=\alpha _{1}^{(j)}E_{1}+\alpha _{2}^{(j)}E_{2}+c_{12\bar{1}}^{(j)}E_{1}E_{2}\bar{E}_{1}+c_{12\bar{2}}^{(j)}E_{1}E_{2}\bar{E}_{2} ,\qquad j=1,2, \\
 h_{n}^{(1)}=1+d_{1\bar{1}}^{(1)}E_{1}\bar{E}_{1}+d_{2\bar{1}}^{(1)}E_{2}\bar{E}_{1}+d_{1\bar{2}}^{(1)}E_{1}\bar{E}_{2}+d_{2\bar{2}}^{(1)}E_{2}\bar{E}
_{2}+d_{12\bar{1}\bar{2}}^{(1)}E_{1}E_{2}\bar{E}_{1}\bar{E}_{2} ,
\end{gather*}
where
\begin{gather*}
c_{i\bar{j}}=\frac{\sum\limits_{k=1}^{2}\bar{\alpha}_{i}^{(k)}\sigma _{k}\alpha
_{j}^{(k)}}{({p}_{i}\bar{p}_{j}-1)^{2}\left(\frac{s}{{p}_{i}\bar{p}_{j}}+\frac{{\sigma _{3}|\rho
_{1}|^{2}}(a_{1}-\bar{a}_{1})^{2}}{|1-a_{1}p_{i}|^{2}|1-\bar{a}_{1}p_{j}|^{2}}\right)} ,\qquad
d_{i\bar{j}}^{(1)}=-\frac{(p_{i}-a_{1})(\bar{p}_{j}-a_{1})}{(1-a_{1}p_{i})(1-a_{1}\bar{p}_{j})}c_{i\bar{j}} ,\\
c_{12\bar{1}\bar{2}}=|p_{2}-p_{1}|^{2}\left( \frac{c_{1\bar{1}}c_{2\bar{2}}}{(p_{1}+\bar{p}_{2}) ( p_{2}+\bar{p}_{1} ) }-\frac{c_{1\bar{2}}c_{2\bar{1}}}{(p_{1}+\bar{p}_{1}) ( p_{2}+\bar{p}_{2} ) }\right) ,\\
c_{12\bar{j}}= ( p_{2}-p_{1} ) \left( \frac{\alpha _{2}^{(1)}c_{1\bar{j}}}{p_{2}+\bar{p}_{j}}-\frac{\alpha _{1}^{(1)}c_{2\bar{j}}}{p_{1}+\bar{p}_{j}}\right) , \\
d_{12\bar{1}\bar{2}}^{(1)} =\frac{(p_{1}-a_{1})(\bar{p}_{1}-a_{1})(p_{2}-a_{2})(\bar{p}_{1}-a_{2})}{(1-a_{1}p_{1}) (1-a_{1}\bar{p}_{1})(1-a_{1}p_{2})(1-a_{1}\bar{p}_{2})}c_{12\bar{1}\bar{2}}.
\end{gather*}

It is found that two-bright-one-dark soliton solution given above is nonsingular for any other combinations of nonlinearities if the following quantity
\begin{gather*}
\left( \sum_{k=1}^{2}\bar{\alpha}_{i}^{(k)}\sigma _{k}\alpha_{j}^{(k)}\right) \left(\frac{s}{|{p}_{i}|^2}+\frac{{\sigma _{3}|\rho_{1}|^{2}}(a_{1}-\bar{a}_{1})^{2}}{|1-a_{1}p_{i}|^{2}|1-\bar{a}_{1}p_{i}|^{2}}\right) %\label{constraint2-2b1d-3CNLS}
\end{gather*}
is positive.

{\bf{One-bright--two-dark soliton solution}.} The tau-functions for one-soliton solution \linebreak \smash{($N=1$)} are
\begin{gather*}
f_{n}=-1-c_{1\bar{1}}(p_{1}\bar{p}_{1})^{n}e^{\eta _{1}+\bar{\eta}_{1}} ,\qquad
g_{n}^{(1)}=-\alpha _{1}^{(1)}p_{1}{}^{n}e^{\eta _{1}},\qquad
h_{n}^{(l)}=-1-d_{1\bar{1}}^{(l)}(p_{1}\bar{p}_{1})^{n}e^{\eta _{1}+\bar{\eta}_{1}} ,
\end{gather*}
where
\begin{gather*}
c_{1\bar{1}}=\frac{\bar{\alpha}_{1}^{(1)}\sigma _{1}\alpha _{1}^{(1)}}{({p}_{1}\bar{p}_{1}-1)^{2}\left(\frac{s}{|p_1|^2}+\frac{\sum\limits_{l=1}^{2}{\sigma _{l+1}|\rho _{l}|^{2}}(a_{l}-\bar{a}_{l})^{2}}{|1-a_{1}p_{1}|^{2}|1-\bar{a}_{1}p_{1}|^{2}}\right)} .
\end{gather*}
Here $\omega _{l}=s(a_{l}-\bar{a}_{l})$, $|a_{l}|=1$, $s=1+\sigma _{2}|\rho _{1}|^{2}+\sigma _{3}|\rho _{2}|^{2}$, $\bar{a}_{l}$ represents the complex conjugate of~$a_{l}$. The above tau functions lead to the one-soliton as follows
\begin{gather*}
q_{n}^{(1)}=\frac{\alpha _{1}^{(1)}}{2}\sqrt{c_{1\bar{1}}}e^{\mathrm{i}\xi_{1I}}\operatorname{sech} ( \xi _{1R}+\theta _{0}) , \\
q_{n}^{(l+1)}=\frac{1}{2}\rho _{l}e^{\mathrm{i}\zeta _{l}}\big( 1+e^{2\mathrm{i}\phi _{l}}+(e^{2\mathrm{i}\phi _{l}}-1)\tanh ( \xi
_{1R}+\theta _{0} ) \big) , \qquad l=1,2,
\end{gather*}
where $\xi _{1}=\xi _{1R}+\mathrm{i}\xi _{1I}=n\ln (\mathrm{i}p_{1})+s\big(p_{1}-p_{1}^{-1}\big)t$, $\zeta _{l}=n\varphi _{l}+n\pi /2+\omega _{l}t$, $e^{\mathrm{i}\varphi _{l}}=a_{l}$, $e^{2\theta _{0}}=c_{1\bar{1}}$, $e^{2\mathrm{i}\phi _{l}} =-(p_{1}-a_{l}) (\bar{p}_{1}-a_{l})/((1-a_{l}p_{1})(1-a_{l}\bar{p}_{1}))$. Therefore, the amplitude of bright soliton for~$q^{(1)}$ are $\frac{1}{2}\big|\alpha _{1}^{(1)}\big|/\sqrt{a_{1\bar{1}}}$. The dark soliton $q^{(l+1)}$ approaches $|\rho _{l}|$ as $x\rightarrow \pm \infty $. In addition, the intensity of the dark soliton is $|\rho _{l}|\cos \phi _{l}$.

The tau functions for two-soliton are of the following form
\begin{gather*}
f_{n}=1+c_{1\bar{1}}E_{1}\bar{E}_{1}+c_{2\bar{1}}E_{2}\bar{E}_{1}+c_{1\bar{2}}E_{1}\bar{E}_{2}+c_{2\bar{2}}E_{2}\bar{E}_{2}+c_{12\bar{1}\bar{2}}E_{1}E_{2}
\bar{E}_{1}\bar{E}_{2} , \\
g_{n}^{(1)}=\alpha _{1}^{(1)}E_{1}+\alpha _{2}^{(1)}E_{2}+c_{12\bar{1}}^{(1)}E_{1}E_{2}\bar{E}_{1}+c_{12\bar{2}}^{(1)}E_{1}E_{2}\bar{E}_{2} , \\
h_{n}^{(l)}=1+d_{1\bar{1}}^{(l)}E_{1}\bar{E}_{1}+d_{2\bar{1}}^{(l)}E_{2}\bar{E}_{1}+d_{1\bar{2}}^{(l)}E_{1}\bar{E}_{2}+d_{2\bar{2}}^{(l)}E_{2}\bar{E}
_{2}+d_{12\bar{1}\bar{2}}^{(l)}E_{1}E_{2}\bar{E}_{1}\bar{E}_{2} , \qquad l=1,2 ,
\end{gather*}
where
\begin{gather*}
c_{i\bar{j}}=\frac{\bar{\alpha}_{i}^{(1)}\sigma _{1}\alpha _{j}^{(1)}}{({p}_{i}\bar{p}_{j}-1)^{2}\left(\frac{s}{{p}_{i}\bar{p}_{j}}+ \frac{\sum\limits_{l=1}^{2}{\sigma }_{l+1}{|\rho _{l}|^{2}}(a_{l}-\bar{a}_{l})^{2}}
{|1-a_{l}p_{i}|^{2}|1-\bar{a}_{l}p_{j}|^{2}}\right)} ,\qquad
d_{i\bar{j}}^{(l)}=-\frac{(p_{i}-a_{l})(\bar{p}_{j}-a_{l})}{(1-a_{l}p_{i})(1-a_{l}\bar{p}_{j})}c_{i\bar{j}} ,\\
c_{12\bar{1}\bar{2}}=|p_{2}-p_{1}|^{2} \left( \frac{c_{1\bar{1}}c_{2\bar{2}}}{(p_{1}+\bar{p}_{2})( p_{2}+\bar{p}_{1}) } -\frac{c_{1\bar{2}}c_{2\bar{1}}}{(p_{1}+\bar{p}_{1})( p_{2}+\bar{p}_{2}) }\right) ,\\
c_{12\bar{j}} =( p_{2}-p_{1}) \left( \frac{\alpha _{2}^{(1)}c_{1\bar{j}}}{p_{2}+\bar{p}_{j}}-\frac{\alpha _{1}^{(1)}c_{2\bar{j}}}{p_{1}+\bar{p}_{j}}\right) , \\
d_{12\bar{1}\bar{2}}^{(l)} =\frac{(p_{1}-a_{l})(\bar{p}_{1}-a_{l})(p_{2}-a_{l})(\bar{p}_{1}-a_{l})}{(1-a_{l}p_{1}) (1-a_{l}\bar{p}_{1})(1-a_{l}p_{2})(1-a_{l}\bar{p}_{2})}c_{12\bar{1}\bar{2}}.
\end{gather*}

For all possible combinations of mixed type in three-coupled NLS equation, one-bright-two-dark soliton solution exists if the following quantity
\begin{gather*}
\big(\bar{\alpha}_{i}^{(1)}\sigma _{1}\alpha _{i}^{(1)}\big) \left(\frac{s}{|{p}_{i}|^{2}}+\frac{\sum\limits_{l=1}^{2}{\sigma }_{l+1}{|\rho _{l}|^{2}}(a_{l}-\bar{a}_{l})^{2}}{|1-a_{l}p_{i}|^{2}|1-\bar{a}_{l}p_{i}|^{2}}\right) %\label{constraint2-1b2d-3CNLS}
\end{gather*}
is positive.

The asymptotic analysis for two-soliton solution is performed in Appendix~\ref{appendixA}. It should be pointed out that two-soliton for one-bright-two-dark soliton case always undertakes elastic collision without shape changing.

\section{Discussion and conclusion}\label{section4}
{We conclude the present paper by two comments. First, we comment on a connection of the vector semi-discrete NLS equation to the vector modif\/ied Volterra lattice equation studied in~\cite{Adler08}. To this end, we consider the two-component semi-discrete NLS equation of focusing type
\begin{gather}
 \mathrm{i} \frac{{\rm d}}{{\rm d}t}u_n=\big(1+|u_n|^2+|v_n|^2\big) (u_{n+1}+u_{n-1} ),\nonumber \\
\mathrm{i} \frac{{\rm d}}{{\rm d}t}v_n=\big(1+|u_n|^2+|v_n|^2\big)(v_{n+1}+v_{n-1}) .\label{2NLS}
\end{gather}
Let $u_n=x_n+\mathrm{i} y_n$ and $v_n=z_n+\mathrm{i} w_n$, then equation~(\ref{2NLS}) becomes
\begin{gather*} %\label{2NLSb}
\frac{{\rm d}}{{\rm d}t}x_n=\big(1+x_n^2+y_n^2+z_n^2+w_n^2\big)(y_{n+1}+y_{n-1}), \\
-\frac{{\rm d}}{{\rm d}t}y_n=\big(1+x_n^2+y_n^2+z_n^2+w_n^2\big)(x_{n+1}+x_{n-1}), \\
\frac{{\rm d}}{{\rm d}t}z_n=\big(1+x_n^2+y_n^2+z_n^2+w_n^2\big)(w_{n+1}+w_{n-1}), \\
-\frac{{\rm d}}{{\rm d}t}w_n=\big(1+x_n^2+y_n^2+z_n^2+w_n^2\big)(z_{n+1}+z_{n-1}) .
\end{gather*}
By def\/ining
\begin{gather*}
 U_n^{(1)}=\begin{cases} x_n, & \text{for $n$ even},\\
 y_n, & \text{for $n$ odd}, \end{cases} \qquad
 U_n^{(2)}=\begin{cases} -y_n, &\text{for $n$ even},\\
 x_n, & \text{for $n$ odd},\end{cases} \\
 U_n^{(3)}=\begin{cases} z_n, &\text{for $n$ even},\\ w_n, &\text{for $n$ odd},\end{cases}\qquad
 U_n^{(4)}=\begin{cases} -w_n, & \text{for $n$ even},\\ z_n, & \text{for $n$ odd},\end{cases} \\
 U_n^{(5)}=\begin{cases} (-1)^{n/2}, &\text{for $n$ even},\\ (-1)^{(n-1)/2}, & \text{for $n$ odd} ,\end{cases}
\end{gather*}
we obtain
\begin{gather*} \frac{{\rm d}}{{\rm d}t}U_n^{(1)}
=\big(\big(U_n^{(1)}\big)^2+\big(U_n^{(2)}\big)^2+\big(U_n^{(3)}\big)^2+\big(U_n^{(4)}\big)^2+\big(U_n^{(5)}\big)^2\big)
\big(U_{n+1}^{(1)}+U_{n-1}^{(1)}\big), \\
 \frac{{\rm d}}{{\rm d}t}U_n^{(2)}
=\big(\big(U_n^{(1)}\big)^2+\big(U_n^{(2)}\big)^2+\big(U_n^{(3)}\big)^2+\big(U_n^{(4)}\big)^2+\big(U_n^{(5)}\big)^2\big)
\big(U_{n+1}^{(2)}+U_{n-1}^{(2)}\big), \\
 \frac{{\rm d}}{{\rm d}t}U_n^{(3)}
=\big(\big(U_n^{(1)}\big)^2+\big(U_n^{(2)}\big)^2+\big(U_n^{(3)}\big)^2+\big(U_n^{(4)}\big)^2+\big(U_n^{(5)}\big)^2\big)
\big(U_{n+1}^{(3)}+U_{n-1}^{(3)}\big), \\
 \frac{{\rm d}}{{\rm d}t}U_n^{(4)}
=\big(\big(U_n^{(1)}\big)^2+\big(U_n^{(2)}\big)^2+\big(U_n^{(3)}\big)^2+\big(U_n^{(4)}\big)^2+\big(U_n^{(5)}\big)^2\big)
\big(U_{n+1}^{(4)}+U_{n-1}^{(4)}\big), \\
 \frac{{\rm d}}{{\rm d}t}U_n^{(5)}
=\big(\big(U_n^{(1)}\big)^2+\big(U_n^{(2)}\big)^2+\big(U_n^{(3)}\big)^2+\big(U_n^{(4)}\big)^2+\big(U_n^{(5)}\big)^2\big)
\big(U_{n+1}^{(5)}+U_{n-1}^{(5)}\big)
\end{gather*}
for $n$ being even, and
\begin{gather*}
 \frac{{\rm d}}{{\rm d}t}U_n^{(1)}
=-\big(\big(U_n^{(1)}\big)^2+\big(U_n^{(2)}\big)^2+\big(U_n^{(3)}\big)^2+\big(U_n^{(4)}\big)^2+\big(U_n^{(5)}\big)^2\big)
\big(U_{n+1}^{(1)}+U_{n-1}^{(1)}\big), \\
 \frac{{\rm d}}{{\rm d}t}U_n^{(2)}
=-\big(\big(U_n^{(1)}\big)^2+\big(U_n^{(2)}\big)^2+\big(U_n^{(3)}\big)^2+\big(U_n^{(4)}\big)^2+\big(U_n^{(5)}\big)^2\big)
\big(U_{n+1}^{(2)}+U_{n-1}^{(2)}\big), \\
 \frac{{\rm d}}{{\rm d}t}U_n^{(3)}
=-\big(\big(U_n^{(1)}\big)^2+\big(U_n^{(2)}\big)^2+\big(U_n^{(3)}\big)^2+\big(U_n^{(4)}\big)^2+\big(U_n^{(5)}\big)^2\big)
\big(U_{n+1}^{(3)}+U_{n-1}^{(3)}\big), \\
 \frac{{\rm d}}{{\rm d}t}U_n^{(4)}
=-\big(\big(U_n^{(1)}\big)^2+\big(U_n^{(2)}\big)^2+\big(U_n^{(3)}\big)^2+\big(U_n^{(4)}\big)^2+\big(U_n^{(5)}\big)^2\big)
\big(U_{n+1}^{(4)}+U_{n-1}^{(4)}\big), \\
 \frac{{\rm d}}{{\rm d}t}U_n^{(5)}
=-\big(\big(U_n^{(1)}\big)^2+\big(U_n^{(2)}\big)^2+\big(U_n^{(3)}\big)^2+\big(U_n^{(4)}\big)^2+\big(U_n^{(5)}\big)^2\big)
\big(U_{n+1}^{(5)}+U_{n-1}^{(5)}\big)
\end{gather*}
for $n$ being odd, in other words,
\begin{gather*}
\frac{{\rm d}}{{\rm d}t}U_n^{(j)}=(-1)^n\left(\sum_{k=1}^5\big(U_n^{(k)}\big)^2\right)\big(U_{n+1}^{(j)}+U_{n-1}^{(j)}\big)
\end{gather*}
for $1\le j\le 5$. By rewriting $U_n^{(j)}=\mathrm{i}^nV_n^{(j)}$ and $t=x/\mathrm{i}$, we have
\begin{gather*} \frac{{\rm d}}{{\rm d}x}V_n^{(j)} =\left(\sum_{k=1}^5\big(V_n^{(k)}\big)^2\right)\big(V_{n+1}^{(j)}-V_{n-1}^{(j)}\big) .
\end{gather*}
Consequently,
\begin{gather*}
 u_n =
\begin{cases}
 U_n^{(1)}-\mathrm{i}U_n^{(2)}, & \text{$n$ even}, \\
U_n^{(2)}+\mathrm{i}U_n^{(1)}, & \text{$n$ odd}
 \end{cases} =
\begin{cases}
\mathrm{i}^{n-1}\big(V_n^{(2)}+\mathrm{i}V_n^{(1)}\big), & \text{$n$ even}, \\
 \mathrm{i}^n\big(V_n^{(2)}+\mathrm{i}V_n^{(1)}\big), & \text{$n$ odd},
 \end{cases}\\
v_n=
\begin{cases} U_n^{(3)}-\mathrm{i}U_n^{(4)}, & \text{$n$ even}, \\
U_n^{(4)}+\mathrm{i}U_n^{(3)}, & \text{$n$ odd}
 \end{cases} =
\begin{cases}
\mathrm{i}^{n-1}\big(V_n^{(4)}+\mathrm{i} V_n^{(3)}\big), & \text{$n$ even}, \\
 \mathrm{i}^n\big(V_n^{(4)}+\mathrm{i}V_n^{(3)}\big), & \text{$n$ odd}.
 \end{cases}
\end{gather*}
$u_n$ and $u_n^*$ correspond to $U_n^{(1)}$ and $U_n^{(2)}$, and $v_n$ and $v_n^*$ correspond to $U_n^{(3)}$ and $U_n^{(4)}$, with even-odd parity depending gauge factor $\mathrm{i}^{\text{($n$ or $n-1$)}}$.

The correspondence between coupled defocusing-defocusing and focusing-defocusing coupled Ablowitz--Ladik equation and the vector modif\/ied Volterra lattice equatino can be constructed by similar variable transformations. In all cases, we need 5-components in the vector modif\/ied Volterra lattice, one of which is a trivial wave $V_n^{(5)}=\mathrm{i}^{\text{($0$ or $1$)}}$, in order to recover 2-component coupled Ablowitz--Ladik equation. In principle any solutions of coupled Ablowitz--Ladik equation can be derived from those of vector modif\/ied Volterra lattice and vice versa. However it is not so easy to make exact matching.}

Second, we give a comparison between the NLS-type equations and the sine/sinh-Gordon-type equations since their belong to the simplest positive and negative f\/lows of the AKNS hierarchy, respectively. In a series papers by Barashenlov, Getmanov et al.~\cite{Getmanov87,Getmanov88,Getmanov93JMP1,Getmanov93JMP2}, a~generic integrable relativistic system associated with the ${\mathfrak{sl}}(2,{\mathbb C})$ was systematically investigated, from which the massive Thirring model, the complex sine-Gordon equation in Euclidean and Minkowski spaces etc.\ are produced by dif\/ferent reductions. Especially an ${\rm O}(1,1)$ sine-Gordon equation with the Lagrangian
\begin{gather*}%\label{Lag_sG}
 L= \frac{u_{1\xi} u_{1\eta}-u_{2\xi} u_{2\eta}}{1-\big(u^2_1-u^2_2\big)}+\big(u^2_1-u^2_2-1\big) ,
\end{gather*}
exhibits nontrivial interaction of solitons such as decay and fusion \cite{Getmanov88}. Here $u_1$ and $u_2$ are real variables. In parallel to the ${\rm O}(1,1)$ sine-Gordon equation, we can have a NLS equation of~${\rm O}(1,1)$ type, whose Lagrangian is
\begin{gather*}%\label{Lag_RNLS}
 L={u_{1} u_{2t}-u_{1t} u_{2}} + \big(u^2_{1x}-u^2_{2x}\big)- \sigma \big(u^2_1-u^2_2\big)^2 .
\end{gather*}

On the contrary, in parallel to the ${\rm U}(1,1)$ coupled NLS system, e.g., system~(\ref{CNLS}) with $\sigma=1$ and $\sigma=-1$, whose Lagrangian can be written by
\begin{gather*}%\label{Lag_CNLS}
 L=\frac{\mathrm{i}}{2} \big(u_{t} u^{\ast}-u^{\ast} u_{t}-v_{t} v^{\ast}+v^{\ast} v_{t}\big) + \big(|u|^2_{x}-|v|^2_{x}\big)- \sigma \big(|u|^2-|v|^2\big)^2,
\end{gather*}
where $u$ and $v$ are complex variable and $^{\ast}$ represents complex conjugate, we could propose a~${\rm U}(1,1)$ coupled complex sine-Gordon system with Lagrangian
\begin{gather*}%\label{Lag_UcsG}
 L= \frac{u_{\eta} u^{\ast}_{\xi}-v_{\eta} v^{\ast}_{\xi}}{1-\big(|u|^2-|v|^2\big)}-\big(|u|^2-|v|^2\big) .
\end{gather*}
Then several natural questions arise: does the NLS equation of ${\rm O}(1,1)$ type exhibit nontrivial interaction of solitons such as decay and fusion? Are there any nontrivial interaction of solitons for the ${\rm U}(1,1)$ coupled NLS system and the ${\rm U}(1,1)$ complex sine-Gordon system? Unfortunately, the answers to these questions are not clear at this moment. We expect above questions could be answered by the authors or others in the near future.

\appendix

\section{Appendix} \label{appendixA}
By taking $N=1$ we get the tau functions for one-soliton solution,
\begin{gather*}
f_{n}=\operatorname{Pf} (a_{1},a_{2},b_{1},b_{2})=-1-c_{1\bar{1}}E_{1}\bar{E}_{1},\\
g_{n}^{(j)}=\operatorname{Pf} (d_{0},\beta _{j},a_{1},a_{2},b_{1},b_{2})=-\alpha _{1}^{(1)}E_{1},\qquad j=1,\dots, m ,\\
h_{n}^{(l)}=\operatorname{Pf} \big(c_{1}^{(l)},c_{2}^{(l)},b_{1},b_{2}\big)=-1-d_{1\bar{1}}^{(l)}E_{1}\bar{E}_{1},\qquad l=1,\dots, M-m ,
\end{gather*}
where
\begin{gather*}
c_{1\bar{1}}=\frac{\sum\limits_{k=1}^{m}\bar{\alpha}_{1}^{(k)}\sigma _{k}\alpha
_{1}^{(k)}}{({p}_{1}\bar{p}_{1}-1)^{2}\left(\frac{s}{|p_1|^2}+\sum\limits_{l=1}^{M-m} \frac{|\rho _{l}|^{2}(a_{l}-\bar{a}%
_{l})^{2}}{|1-a_{l}p_{1}|^{2}|1-\bar{a}_{l}p_{1}|^{2}}\right)} ,\qquad
d_{1\bar{1}}^{(l)}=-\frac{(p_{1}-a_{l})(\bar{p}_{1}-a_{l})}{(1-a_{l}p_{1})(1-a_{l}\bar{p}_{1})}c_{1\bar{1}}.
\end{gather*}

Based on the $N$-soliton solution of the vector discrete NLS equation, the tau-functions for two-soliton solution can be expanded for $N=2$
\begin{gather*}
f=\operatorname{Pf} (a_{1},a_{2},a_{3},a_{4},b_{1},b_{2},b_{3},b_{4}) \\
\hphantom{f}{} =1+c_{1\bar{1}}E_{1}\bar{E}_{1}+c_{2\bar{1}}E_{2}\bar{E}_{1}+c_{1\bar{2}}E_{1} \bar{E}_{2}+c_{2\bar{2}}E_{2}\bar{E}_{2}+c_{12\bar{1}\bar{2}}E_{1}E_{2}\bar{E}_{1}\bar{E}_{2},
\\
 g_{n}^{(j)} =\operatorname{Pf} (d_{0},\beta_{j},a_{1},a_{2},a_{3},a_{4},b_{1},b_{2},b_{3},b_{4}) \\
\hphantom{g_{n}^{(j)}}{} =\alpha _{1}^{(j)}E_{1}+\alpha _{2}^{(j)}E_{2}+c_{12\bar{1}}^{(j)}E_{1}E_{2}\bar{E}_{1}+c_{12\bar{2}}^{(j)}E_{1}E_{2}\bar{E}_{2} , \qquad j=1,\dots, m ,\\
 h_{n}^{(l)} =\operatorname{Pf} \big(c_{1}^{(l)},c_{2}^{(l)},c_{3}^{(l)},c_{4}^{(l)},b_{1},b_{2},b_{3},b_{4}\big) =1+d_{1\bar{1}}^{(l)}E_{1}\bar{E}_{1} \\
\hphantom{h_{n}^{(l)}=}{}+d_{2\bar{1}}^{(l)}E_{2}\bar{E}_{1} +d_{1\bar{2}}^{(l)}E_{1}\bar{E}_{2}+d_{2\bar{2}}^{(l)}E_{2}\bar{E}_{2}+d_{12\bar{1}\bar{2}}^{(l)}E_{1}E_{2}\bar{E}_{1}\bar{E}_{2} , \qquad l=1,\dots, M-m ,
\end{gather*}
where
\begin{gather*}
c_{i\bar{j}}=\frac{\sum\limits_{k=1}^{m}\bar{\alpha}_{j}^{(k)}\sigma
_{k}\alpha _{i}^{(k)}}{({p}_{i}\bar{p}_{j}-1)^{2}\left(\frac{s}{{p}_{i}\bar{p}_{j}}+\sum\limits_{l=1}^{M-m} \frac{\rho _{l}|^{2}(a_{l}-\bar{a}
_{l})^{2}}{|1-a_{l}p_{i}|^{2}|1-\bar{a}_{l}p_{j}|^{2}}\right)} ,\qquad
d_{i\bar{j}}^{(l)}=-\frac{(p_{i}-a_{l})(\bar{p}_{j}-a_{l})}{(1-a_{l}p_{i})(1-a_{l}\bar{p}_{j})}c_{i\bar{j}} ,\\
c_{12\bar{1}\bar{2}}=|P_{12}|^{2}\big( P_{1\bar{1}}P_{2\bar{2}}c_{1\bar{2}}c_{2\bar{1}}-P_{1\bar{2}}P_{2\bar{1}}c_{1\bar{1}}c_{2\bar{2}}\big),\qquad
c_{12\bar{j}}^{(k)} =P_{12}\big( \alpha _{1}^{(k)}P_{1\bar{1}}c_{2\bar{j}}-\alpha _{2}^{(k)}P_{2\bar{1}}c_{1\bar{j}}\big) , \\
d_{12\bar{1}\bar{2}}^{(l)} =\frac{(p_{1}-a_{l})(\bar{p}_{1}-a_{l})(p_{2}-a_{l})(\bar{p}_{1}-a_{l})}{(1-a_{l}p_{1}) (1-a_{l}\bar{p}_{1})(1-a_{l}p_{2})(1-a_{l}\bar{p}_{2})}c_{12\bar{1}\bar{2}},\\
P_{ij}=\frac{p_{i}-p_{j}}{p_{i}p_{j}-1} ,\qquad P_{i\bar{j}}=\frac{p_{i}-\bar{p}_{j}}{p_{i}\bar{p}_{j}-1} .
\end{gather*}
Next, we investigate the asymptotic behavior of two-soliton solution. To this end, we assume $(\operatorname{Re}\ln (\mathrm{i}p_{2}))>\operatorname{Re} (\ln (\mathrm{i}p_{1}))>0$, $s\operatorname{Re}(p_{2}-p_{2}^{-1})/\operatorname{Re}(\ln (\mathrm{i}p_{2}))>s\operatorname{Re}(p_{1}-p_{1}^{-1})/\operatorname{Re}(\ln (\mathrm{i}p_{1}))$ without loss of generality. For the above choice of parameters, we have (i)~$\eta _{1R}\approx 0$, $\eta _{2R}\rightarrow \mp \infty $ as $t\rightarrow \mp \infty $ for soliton~1 and (ii)~$\eta _{2R}\approx 0$, $\eta _{1R}\rightarrow \pm \infty $ as $t\rightarrow \mp \infty $ for soliton~2. This leads to the following asymptotic forms for two-soliton solution.

(i) Before collision ($t\rightarrow -\infty $). Soliton 1 ($\eta _{1R}\approx 0$, $\eta _{2R}\rightarrow -\infty $):
\begin{gather*}
q_{n}^{(j)} \rightarrow \alpha _{1}^{(j)}\frac{\mathrm{i}^{n}E_{1}}{1+c_{1\bar{1}}E_{1}\bar{E}_{1}} \rightarrow A_{j}^{1-}e^{\mathrm{i}\xi _{1I}}\operatorname{sech} \big( \xi _{1R}+\xi_{0}^{1-}\big) ,\\
q_{n}^{(l+m)} \rightarrow \frac{1+d_{1\bar{1}}^{(l)}E_{1}\bar{E}_{1}}{1+c_{1\bar{1}}E_{1}\bar{E}_{1}}\rho _{l} (\mathrm{i}a_{l})^{n}e^{\omega _{l}t} \\
\hphantom{q_{n}^{(l+m)}}{} \rightarrow \frac{1}{2}\rho _{l}e^{\mathrm{i}\zeta _{l}}\big( 1+e^{2\mathrm{i}\phi _{l}^{(1)}}+\big(e^{2\mathrm{i}\phi _{l}^{(1)}}-1\big)\tanh \big(\xi _{1R}+\xi _{0}^{1-}\big) \big),
\end{gather*}
where
\begin{gather*}
A_{j}^{1-}=\frac{\alpha _{1}^{(j)}}{2\sqrt{c_{1\bar{1}}}} , \qquad e^{2\xi _{0}^{1-}}=c_{1\bar{1}},\\
e^{2\mathrm{i}\phi _{l}^{(1)}} =-\frac{(p_{1}-a_{l})(\bar{p}_{1}-a_{l})}{(1-a_{l}p_{1})(1-a_{l}\bar{p}_{1})},\qquad \zeta _{l}=n\varphi _{l}+n\pi /2+\omega_{l}t,\\
\xi _{1} =\xi _{1R}+\mathrm{i}\xi _{1I}=n\ln (\mathrm{i}p_{1})+s\big(p_{1}-p_{1}^{-1}\big)t,\qquad e^{\mathrm{i}\varphi _{l}}=a_{l} .
\end{gather*}

Soliton 2 ($\eta _{2R}\approx 0$, $\eta _{1R}\rightarrow \infty $):
\begin{gather*}
q_{n}^{(j)} \rightarrow \frac{\mathrm{i}^{n}c_{12\bar{1}}^{(j)}E_{2}}{c_{1\bar{1}}+c_{12\bar{1}\bar{2}}E_{2}\bar{E}_{2}} \rightarrow A_{j}^{2-}e^{\mathrm{i}\xi _{2I}}\operatorname{sech} \big( \xi _{2R}+\xi_{0}^{2-}\big) ,\\
q_{n}^{(l+m)} \rightarrow \frac{d_{1\bar{1}}^{(l)}+d_{12\bar{1}\bar{2}}^{(l)}E_{2}\bar{E}_{2}}{c_{1\bar{1}}+c_{12\bar{1}\bar{2}}E_{2}
\bar{E}_{2}}\rho _{l} (\mathrm{i}a_{l})^{n}e^{\omega _{l}t}\\
\hphantom{q_{n}^{(l+m)}}{}
\rightarrow \frac{1}{2}\rho _{l}e^{\mathrm{i}\big(\zeta _{l}+2\phi_{l}^{(1)}\big)}\big( 1+e^{2\mathrm{i}\phi _{l}^{(2)}}+\big(e^{2\mathrm{i}\phi
_{l}^{(2)}}-1\big)\tanh \big( \xi _{2R}+\xi _{0}^{2-}\big) \big),
\end{gather*}
where
\begin{gather*}
A_{j}^{2-}=\frac{c_{12\bar{1}}^{(j)}}{2\sqrt{c_{1\bar{1}}}\sqrt{c_{12\bar{1}\bar{2}}}} , \qquad e^{2\xi _{0}^{2-}}=\frac{c_{12\bar{1}\bar{2}}}{c_{1\bar{1}}},\\
e^{2\mathrm{i}\phi _{l}^{(2)}} = -\frac{(p_{2}-a_{l})(\bar{p}_{2}-a_{l})}{(1-a_{l}p_{2})(1-a_{l}\bar{p}_{2})}, \qquad \zeta _{0}^{2-}=n\varphi _{l}+n\pi/2+\omega _{l}t,\\
\xi _{2} =\xi _{2R}+\mathrm{i}\xi _{2I}=n\ln (\mathrm{i}p_{2})+s\big(p_{2}-p_{2}^{-1}\big)t .
\end{gather*}

(ii) After the collision ($t\rightarrow \infty $). Soliton 1 ($\eta _{1R}\approx 0$, $\eta _{2R}\rightarrow +\infty $):
\begin{gather*}
q_{n}^{(j)} \rightarrow \frac{c_{12\bar{2}}^{(j)}\mathrm{i}^{n}E_{1}}{c_{2\bar{2}}+c_{12\bar{1}\bar{2}}E_{1}\bar{E}_{1}} \rightarrow A_{j}^{1+}e^{\mathrm{i}\xi _{1I}} \operatorname{sech}\big( \xi _{1R}+\xi_{0}^{1+}\big) ,
\\
q_{n}^{(l+m)} \rightarrow \frac{d_{2\bar{2}}^{(l)}+d_{12\bar{1}\bar{2}}^{(l)}E_{1}\bar{E}_{1}}{c_{2\bar{2}}+c_{12\bar{1}
\bar{2}}E_{1}\bar{E}_{1}}\rho _{l} (\mathrm{i}a_{l})^{n}e^{\omega
_{l}t} \\
\hphantom{q_{n}^{(l+m)} }{} \rightarrow \frac{1}{2}\rho _{l}e^{\mathrm{i}\big(\zeta _{l}+2\mathrm{i}\phi
_{l}^{(2)}\big)}\big( 1+e^{2\mathrm{i}\phi _{l}^{(1)}}+\big(e^{2\mathrm{i}\phi
_{l}^{(1)}}-1\big)\tanh \big( \xi _{1R}+\xi _{0}^{1+}\big) \big),
\end{gather*}
 where
\begin{gather*}
A_{j}^{1+}=\frac{c_{12\bar{2}}^{(j)}}{2\sqrt{c_{1\bar{1}}}\sqrt{c_{12\bar{1}\bar{2}}}} ,\qquad e^{2\xi _{0}^{1+}}=\frac{c_{12\bar{1}\bar{2}}}{c_{2\bar{2}}},
\\
e^{2\mathrm{i}\phi _{l}^{(2)}}=-\frac{(p_{2}-a_{l})(\bar{p}_{2}-a_{l})}{(1-a_{l}p_{2})(1-a_{l}\bar{p}_{2})},\qquad \zeta _{l}=n\varphi _{l}+n\pi /2+2\phi _{l}^{(2)}+\omega _{l}t,
\end{gather*}

Soliton 2 ($\eta _{2R}\approx 0$, $\eta _{1R}\rightarrow -\infty $):
\begin{gather*}
q_{n}^{(j)} \rightarrow \frac{\alpha _{2}^{(j)}\mathrm{i}^{n}E_{2}}{1+c_{2\bar{2}}E_{2}\bar{E}_{2}} \rightarrow A_{j}^{2+}e^{\mathrm{i}\xi _{2I}}\operatorname{sech} \big( \xi _{2R}+\xi_{0}^{2+}\big) ,\\
q_{n}^{(l+m)} \rightarrow \frac{1+d_{2\bar{2}}^{(l)}E_{2}\bar{E}_{2}}{1+c_{2\bar{2}}E_{2}\bar{E}_{2}}\rho _{l} (\mathrm{i}a_{l})^{n}e^{\omega _{l}t}\rightarrow \frac{1}{2}\rho _{l}e^{\mathrm{i}\zeta _{l}}\big( 1+e^{2\mathrm{i}\phi _{l}^{(2)}}+\big(e^{2\mathrm{i}\phi _{l}^{(2)}}-1\big)\tanh \big(\xi _{2R}+\xi _{0}^{2+}\big) \big),
\end{gather*}
where
\begin{gather*}
A_{j}^{2+}=\frac{\alpha _{2}^{(j)}}{2\sqrt{c_{2\bar{2}}}} , \qquad e^{2\xi _{0}^{2+}}=c_{2\bar{2}} .
\end{gather*}

\subsection*{Acknowledgements}
We greatly appreciate all referees' useful comments which help us improve the present paper signif\/icantly.
The work of B.F.\ is partially supported by NSF Grant (No. 1715991) and the COS Research Enhancement Seed Grants Program at UTRGV.
The work of Y.O.\ is partly supported by JSPS Grant-in-Aid for Scientif\/ic Research (B-24340029, S-24224001, C-15K04909) and for
Challenging Exploratory Research (26610029).

\pdfbookmark[1]{References}{ref}
\LastPageEnding

\end{document}